\documentclass[english,11pt]{article}
\usepackage{fullpage}
\usepackage[ruled,vlined]{algorithm2e}
\usepackage[final]{graphics,graphicx}
\usepackage{fixme}
\usepackage{amsmath}
\usepackage{amssymb}
\usepackage{amsfonts}
\usepackage{babel}
\usepackage{graphics}
\usepackage{graphicx}
\usepackage{subfigure}
\usepackage{scalefnt}
\usepackage{environ}

\usepackage{tikz}
\usetikzlibrary{arrows,automata,shapes,decorations,calc,matrix,positioning}
\newtheorem{theorem}{Theorem}

\newtheorem{lemma}[theorem]{Lemma}

\newenvironment{proof}{\noindent {\bf Proof }}{\qed}
\newtheorem{definition}[theorem]{Definition}

\newcommand{\qed}{\penalty 1000 \hfill\penalty 1000$\Box$\par\medskip}

\newcommand{\abs}[1]{\left| {#1} \right|}
\newcommand{\norm}[1]{\left\| {#1} \right\|}

\newcommand{\ee}{{\bf e}}

\DeclareMathOperator*{\argmax}{~\!argmax~\!}

\newcommand{\I}{\ensuremath{{\rm \mathcal I}}}
\newcommand{\R}{\ensuremath{{\rm \mathbb R}}}

\newcommand{\q}{\ensuremath{{\bf q}}}
\newcommand{\cc}{\ensuremath{{\bf c}}}
\newcommand{\x}{\ensuremath{{\bf x}}}
\newcommand{\y}{\ensuremath{{\bf y}}}
\newcommand{\z}{\ensuremath{{\bf z}}}
\newcommand{\vv}{\ensuremath{{\bf v}}}

\newcommand{\w}{\ensuremath{{\bf w}}}
\newcommand{\rr}{\ensuremath{{\bf r}}}
\newcommand{\T}{\mathsf{T}}

\newcommand{\0}{\ensuremath{{\bf 0}}}

\usepackage{ifpdf}
\ifpdf
\else
\fi

\makeatletter

\title{The complexity of interior point methods for solving discounted
  turn-based stochastic games}

\author{
Thomas Dueholm Hansen\thanks{
Department of Computer Science,
Aarhus University, Denmark. E-mail:
{\tt \{tdh, rij\}@cs.au.dk}.
Supported by the Sino-Danish Center for the Theory of Interactive Computation,
funded by the Danish National Research Foundation and the National
Science Foundation of China (under the grant 61061130540); and by the
Center for research in the Foundations of Electronic Markets (CFEM),
supported by the Danish Strategic Research Council.}
\thanks{Thomas Dueholm Hansen is a recipient of the Google Europe
  Fellowship in Game Theory, and this research is supported in part by this
Google Fellowship. Part of the work was performed at the
School of Computer Science, Tel Aviv University, Israel, with the
support of The Danish Council for Independent Research $|$ Natural Sciences (grant no. 12-126512).
}
\and Rasmus Ibsen-Jensen$^*$
}

\date{}

\begin{document}
\maketitle
\begin{abstract}
We study the problem of solving discounted, two player, turn based,
stochastic games (2TBSGs). Jurdzi{\'n}ski and Savani showed that
2TBSGs with deterministic transitions can be reduced to solving
$P$-matrix linear complementarity problems (LCPs). We show that the
same reduction works for general 2TBSGs.
This implies that a number of interior point methods for solving $P$-matrix LCPs can be used to
solve 2TBSGs. We consider two such algorithms. First, we consider
the unified interior point method of Kojima, Megiddo, Noma, and
Yoshise, which runs in time $O((1+\kappa)n^{3.5}L)$, where 
$\kappa$ is a parameter that depends on the $n \times n$ matrix $M$ defining the LCP, and $L$ is the number of bits in the representation of $M$. Second, we consider the interior point potential reduction algorithm
of Kojima, Megiddo, and Ye, which runs in time $O(\frac{-\delta}{\theta}n^4\log \epsilon^{-1})$, where $\delta$ and $\theta$ are parameters that depend on $M$, and $\epsilon$ describes the quality of the solution. For 2TBSGs with
$n$ states and discount factor $\gamma$ we prove that in the worst case
$\kappa = \Theta(n/(1-\gamma)^2)$, $-\delta =
\Theta(\sqrt{n}/(1-\gamma))$, and $1/\theta =
\Theta(n/(1-\gamma)^2)$.
The lower bounds for $\kappa$, $-\delta$, and $1/\theta$ are
obtained using the same family of deterministic games.
\end{abstract}

\section{Introduction}

{\bf  Two-player turn-based stochastic games (2TBSGs).}
A {\em two-player turn-based stochastic game} (2TBSG) consists of a finite set of states and for each state a finite set of actions. The game is played 
by two players (Player~1 and Player~2) for an
infinite number of rounds.  The states are partitioned into two sets $S^1$ 
and $S^2$, belonging to Player 1 and Player 2, respectively.
In each round the game is in some state, and
the player controlling the current state $i$ 
chooses an action $a$ available from state $i$. Every action is associated with a probability distribution over states, and the next state is picked at random according to the probability distribution for $a$. After every transition the game ends with probability $1-\gamma>0$, where
$\gamma$ is the \emph{discount factor} of the game. Every action has
an associated cost, and the objective of Player~1 is to \emph{minimize}
the expected sum of costs, whereas the objective of Player~2 is to \emph{maximize} the expected sum of costs, i.e., the game is a zero-sum game. Our main focus is the case where all states have two actions.

The class of (turn-based) stochastic games was introduced by Shapley~\cite{Sha53} in 1953 and has received much attention since then. 
For books on the subject see, e.g., Neyman and Sorin \cite{NeySor03}
and Filar and Vrieze \cite{FilVri97}. 
2TBSGs with a single player, i.e., one player controls all the states, are known as \emph{Markov decision processes}; a problem that is important in its own right (see, e.g., Puterman \cite{Puterman94}).
Shapley showed that every state in a 2TBSG has a \emph{value} that can be enforced by both players (a property known as \emph{determinacy}). {\em Solving} a 2TBSG means finding the values of all the states, and the problem of solving 2TBSGs is the topic of this paper. 

{\bf Classical algorithms for solving 2TBSGs.} 2TBSGs form an intriguing class of games whose
status in many ways resembles that of linear programming 40 years ago.
They can be solved efficiently with \emph{strategy iteration} algorithms, resembling the simplex method for linear programming, but at the same time no polynomial time
algorithm is known. Strategy iteration algorithms were first described for discounted 2TBSGs by Rao
{\em et al.}~\cite{RCN73}. Building on a result by Ye \cite{Ye11}, Hansen {\em et al.}~\cite{HansenMZ13} showed that the standard strategy iteration algorithm solves 2TBSGs with a
fixed discount, $\gamma$, in \emph{strongly} polynomial time. Prior to this
result a polynomial bound by Littman~\cite{Littman96} was known for
the case where $\gamma$ is fixed. Littman showed that Shapley's \emph{value
  iteration} algorithm \cite{Sha53} solves discounted 2TBSGs in time 
$O(\frac{n m L}{1-\gamma}\log\frac{1}{1-\gamma})$, where $n$
is the number of states, $m$ is the number of actions, and $L$ is the
number of bits needed to represent the game. For a more thorough
introduction to the background of the problem we refer to
Hansen {\em et al.}~\cite{HansenMZ13} and the references therein.

{\bf Interior point methods.} Several polynomial time algorithms have been discovered for solving linear programs. Perhaps the most notable family of such algorithms is {\em interior point methods}. The first interior point method was introduced by Karmarkar~\cite{Kar84} in 1984, and the technique has since been studied extensively
and applied in other contexts. See, e.g., Ye \cite{Ye98}. In particular,
interior point methods can be used to solve \emph{$P$-matrix linear
complementarity problems} (LCPs). One may hope that a polynomial time algorithm for solving 2TBSGs whose discount factor $\gamma$ is not fixed (i.e.,
the discount factor is part of the input) can be obtained through the
use of interior point methods. Indeed, this approach was suggested by
Jurdzi{\'n}ski and Savani \cite{JurdSav08} and Hansen {\em et al.}
\cite{HansenMZ13}. In this paper we study a reduction from 2TBSGs to 
$P$-matrix LCPs, and we study the complexity of two known interior point methods when applied to the resulting $P$-matrix LCPs. 

{\bf $P$-matrix linear complementarity problems.}
A linear complementarity problem (LCP) is defined as follows: Given an
$(n\times n)$-matrix $M$ and a vector $\q\in \R^n$, find two vectors
$\w ,\z\in \R^n$, such that $\w=\q+M\z$, $\w^\T\z=0$, and $\w,\z\geq
\0$. LCPs have also received much attention. For books on the subject
see, e.g., Cottle {\em et al.}~\cite{CoPaSt92} and Ye \cite{Ye98}. 

Jurdzi{\'n}ski and Savani~\cite{JurdSav08} showed that solving
a \emph{deterministic} 2TBSG $G$, i.e., every action leads to a single
state with probability 1, can be reduced to solving an LCP $(M,\q)$. 
G{\"a}rtner and R{\"u}st \cite{GartRust05} gave a similar reduction
from simple stochastic games; a class of games that is polynomially
equivalent to 2TBSGs (see \cite{AnMi09}).
Moreover, Jurdzi{\'n}ski and Savani \cite{JurdSav08}, and G{\"a}rtner
and R{\"u}st \cite{GartRust05}, showed that the resulting matrix $M$ is
a \emph{$P$-matrix} (i.e., all principal sub-matrices have a positive
determinant). We show that the reduction of Jurdzi{\'n}ski and Savani
also works for general 2TBSGs, and that the resulting matrix $M$ is
again a $P$-matrix.

Krishnamurthy {\em et al.}~\cite{KriParRav12} gave a survey on various stochastic games and LCP formulations of those.

{\bf The unified interior point method.}
There exist various interior point methods for solving $P$-matrix
LCPs. One algorithm we consider in this paper is the unified
interior point method of Kojima, Megiddo, Noma, and
Yoshise~\cite{Koj91}. The unified 
interior point method solves an LCP whose matrix $M \in \R^{n\times n}$ is
a \emph{$P_*(\kappa)$-matrix} in time $O((1+\kappa)n^{3.5}L)$, where $L$ is
the number of bits needed to describe $M$. A matrix $M$ is a
$P_*(\kappa)$-matrix, for $\kappa \ge 0$, if 
and only if for all vectors $\x \in \R^n$, we have that
$\x^\T (M\x) + 4\kappa\sum_{i\in \delta_+(M)} \x_i(M\x)_i \geq 0$,
where $\delta_+(M) = \{i \in [n] \mid \x_i(M\x)_i>0\}$.
If $M$ is a $P$-matrix then it
is also a $P_*(\kappa)$-matrix for some $\kappa \ge 0$. Hence,
the algorithm can be used to solve 2TBSGs.

Following the work of Kojima {\em et al.}~\cite{Koj91}, many
algorithms with complexity polynomial in $\kappa$, $L$, and $n$ have
been introduced. For additional examples see \cite{C08,CZZ09,IlNaTe10}.

{\bf An interior point potential reduction algorithm.} 
The second interior point method we consider is the
potential reduction algorithm of Kojima, Megiddo, and Ye
\cite{KoMeYe92}. (See also Ye~\cite{Ye98}.) The potential reduction
algorithm is an interior point method that takes as input a 
$P$-matrix LCP and a parameter $\epsilon > 0$, and produces an
approximate solution $\w,\z$, such that $\w^\T\z<\epsilon$, 
$\w=\q+M\z$, and $\w,\z\geq \0$. The running time of the
algorithm is $O(\frac{-\delta}{\theta} n^{4} \log \epsilon^{-1})$,
where $\delta$ is the least eigenvalue of $\frac{1}{2}(M+M^\T)$, and
$\theta$ is the positive $P$-matrix  number of $M$, that is,
$\theta = \min_{\norm{\x}_2 =1}\max_{i\in\{1,\dots,n\}} \x_i (M\x)_i$.
We refer to $\frac{-\delta}{\theta}$ as the condition number of $M$.
The analysis involving the condition number appears in Ye~\cite{Ye98}.

R{\"u}st~\cite{Rust07} showed that there exists
a simple stochastic game for which the $P$-matrix LCPs resulting from
the reduction of G{\"a}rtner and R{\"u}st \cite{GartRust05} has a large
condition number. The example of R{\"u}st contains a parameter that can
essentially be viewed as the discount factor $\gamma$ for 2TBSGs, and he
shows that the condition number can depend linearly on
$\frac{1}{1-\gamma}$. To be more precise, R{\"u}st~\cite{Rust07} showed
that the matrix $M$ resulting from the reduction of G{\"a}rtner and
R{\"u}st \cite{GartRust05} has positive $P$-matrix number smaller than
1, and that the smallest eigenvalue of the matrix $\frac{1}{2}(M+M^{\T})$
is $-\Omega\big(\frac{1}{1-\gamma}\big)$. This bound can be viewed as
a precursor for our results.

\subsection{Our contributions}
Our contributions are as follows. We show that the reduction by Jurdzi{\'n}ski and Savani~\cite{JurdSav08} from deterministic
2TBSGs to $P$-matrix LCPs generalizes to 2TBSGs without
modification. Although the reduction is the same, we provide an
alternative proof that the resulting matrix is a $P$-matrix.
Let $M_G$ be the
matrix obtained from the reduction for a given 2TBSG $G$.
We also show that for any 2TBSG $G$ with $n$ states and discount factor $\gamma$ we have:
\begin{itemize}
\item[$(i)$] 
The matrix $M_G$ is a $P_*(\kappa)$-matrix for
$\kappa=O\big(\frac{n}{(1-\gamma)^2}\big)$.
\item[$(ii)$] 
The matrix $\frac{M_{G}+M_{G}^{\T}}{2}$ has
smallest eigenvalue $\delta$, where
$-\delta = O\big(\frac{\sqrt{n}}{1-\gamma}\big)$. 
\item[$(iii)$] 
The reciprocal of the positive $P$-matrix number of $M_G$ is $\frac{1}{\theta(M_G)} = O\big(\frac{n}{(1-\gamma)^2}\big)$.
\end{itemize}
Item $(i)$ implies that the running time
of the unified interior point method of Kojima {\em et
  al.}~\cite{Koj91} for 2TBSGs is at most
$O((1+\kappa)n^{3.5}L) = O\big(\frac{n^{4.5}L}{(1-\gamma)^2}\big)$.
Items $(ii)$ and $(iii)$ together imply that the running time of 
the potential reduction algorithm of Kojima {\em et
  al.}~\cite{KoMeYe92} for 2TBSGs is at most $O\big(\frac{-\delta}{\theta} n^{4} \log \epsilon^{-1}\big) = O\big(\frac{n^{5.5} \log 
  \epsilon^{-1}}{(1-\gamma)^3}\big)$. 

Finally, we define a family of \emph{deterministic} 2TBSGs $G_n$ for which we prove matching lower bounds: ~$(i)$~ $\kappa=\Omega\big(\frac{n}{(1-\gamma)^2}\big)$, ~$(ii)$~ $-\delta = \Omega\big(\frac{\sqrt{n}}{1-\gamma}\big)$, ~and ~$(iii)$~ $\frac{1}{\theta(M_{G_n})} = \Omega\big(\frac{n}{(1-\gamma)^2}\big)$.

Note that the upper bounds we prove for the two algorithms are worse than the $O(\frac{n^2 L}{1-\gamma}\log\frac{1}{1-\gamma})$ bound for the value iteration algorithm for the case where every state has two actions. Although our results for existing interior
point methods for solving 2TBSGs are therefore negative, other interior point methods may solve 2TBSGs efficiently. It is also possible that the known bounds for the algorithms studied in this paper do not reflect their true running time.
In fact, we believe that interior point methods are the key to solving 2TBSGs efficiently, and that the study of interior point methods in this context remains an important question for future research.

\subsection{Overview}
In Section~\ref{sec:prelim} we formally introduce the various classes
of problems under consideration. More precisely, in
Subsection~\ref{sec:lcp} we define LCPs, and in
Subsection~\ref{sec:2TBSG} we define 2TBSGs. In
Subsection~\ref{sec:2TBSG2LCP}, we show that the reduction by
Jurdzi{\'n}ski and Savani~\cite{JurdSav08} from {\em deterministic}
2TBSGs to $P$-matrix LCPs generalizes to general 2TBSGs. In
Section~\ref{sec:counter} we estimate the $\kappa$ for which the
matrices of 2TBSGs are $P_*(\kappa)$-matrices, thereby bounding
the running time of the unified interior point
method of Kojima {\em et al.}~\cite{Koj91}. In Section~\ref{sec:ye_alg} we similarly bound the smallest
eigenvalue and the positive $P$-matrix number, proving a bound for
the running time of the potential reduction algorithm of Kojima {\em et
  al.}~\cite{KoMeYe92}.

\section{Preliminaries\label{sec:prelim}}

\subsection{Linear complementarity problems\label{sec:lcp}}

\begin{definition}[\bf Linear complementarity problems]
A \emph{linear complementarity problem} (LCP) is a pair $(M,\q)$,
where $M$ is an $(n \times n)$-matrix and $\q$ is an $n$-vector. A
\emph{solution} to the LCP $(M,\q)$ is a pair of vectors $(\w, \z) \in \R^n$ such
that:
\[
\begin{split}
\w &~=~ \q + M \z \\
\w^\T \z &~=~ 0 \\
\w,\z &~\ge~ \0\enspace.
\end{split}
\]
\end{definition}


\begin{definition}[\bf $P$-matrix]
A matrix $M \in \R^{n \times n}$ is a \emph{$P$-matrix} if and only if
all its principal sub-matrices have a positive determinant.
\end{definition}

The following lemma gives an alternative definition of
$P$-matrices (see, e.g., \cite[Theorem 3.3.4]{CoPaSt92}).

\begin{lemma}\label{lem:alternative}
A matrix $M \in \R^{n \times n}$ is a \emph{$P$-matrix} if and only if
for every non-zero vector $\x \in \R^n$ there is an $i \in [n] = \{1,\dots,n\}$ such that
$\x_i (M\x)_i > 0$.
\end{lemma}

\begin{definition}[\bf Positive $P$-matrix number]\label{def:positive_P_matrix}
The positive $P$-matrix number of a matrix $M \in \R^{n\times n}$ is 
\[
\theta(M) ~=~ \min_{\norm{\x}_2 =1}~\max_{i\in[n]} ~~ \x_i (M\x)_i \enspace. 
\]
\end{definition}

Note that, according to Lemma~\ref{lem:alternative}, $\theta(M) > 0$
if and only if $M$ is a $P$-matrix.

\begin{definition}[\bf $P_*(\kappa)$-matrix]\label{def:kappa}
A matrix $M \in \R^{n\times n}$ is a \emph{$P_*(\kappa)$-matrix}, for $\kappa \ge 0$, if
and only if for every vector $\x \in \R^n$:
\[
\sum_{i\in \delta_-(M)} \x_i(M\x)_i ~+~ (1+4\kappa)\!\!\sum_{i\in \delta_+(M)} \x_i(M\x)_i ~\geq~ 0 \enspace ,
\]
where $\delta_-(M) = \{i \in [n] \mid \x_i(M\x)_i<0\}$ and
$\delta_+(M) = \{i \in [n] \mid \x_i(M\x)_i>0\}$. We say that $M$
is a \emph{$P_*$-matrix} if and only if it is a $P_*(\kappa)$-matrix for some
$\kappa \ge 0$.
\end{definition}

Kojima {\em et al.}~\cite{Koj91} showed that every $P$-matrix is also a
$P_*$-matrix. By definition, a symmetric matrix $M$ is a $P_*(0)$-matrix if and
only if it is \emph{positive semi-definite}. 

\subsection{Two-player turn-based stochastic games\label{sec:2TBSG}}

\begin{definition}[\bf Two-player turn-based stochastic games]
A \emph{two-player turn-based stochastic game} (2TBSG) is a tuple,
$G = (S^1,S^2,(A_i)_{i\in S^1\cup S^2},p,c,\gamma)$, where
\begin{itemize}
\item{}$S^k$, for $k \in \{1,2\}$, is the set of {\em states}
  belonging to Player $k$. We let $S = S^1 \cup S^2$ be the set of all
  states, and we assume that $S^1$ and $S^2$ are disjoint.
\item{}$A_i$, for $i \in S$, is the set of {\em actions} applicable
  from state $i$. We let $A = \bigcup_{i \in S} A_i$ be the set of all
  actions. We assume that $A_i$ and $A_j$ are disjoint for $i \ne
  j$, and that $A_i \ne \emptyset$ for all $i \in S$.

\item
$p: A \to \Delta(S)$ is a map from actions to probability
  distributions over states.
\item{}$c: A \to \R$ is a function that assigns a {\em cost} to every action.
\item{}$\gamma < 1$ is a (positive) {\em discount factor}.
\end{itemize}
\end{definition}

We let $n = |S|$ and $m = |A|$. Furthermore, we let $A^k = \bigcup_{i
  \in S^k} A_i$, for $k \in \{1,2\}$. 
We refer to Player 1 as the minimizer and to Player 2 as the maximizer. 
Figure \ref{fig:example} shows
an example of a simple 2TBSG. The large circles and squares represent
the states controlled by Player 1 and 2, respectively. The edges
leaving the states represent actions. The cost of an action is shown
inside the corresponding diamond shaped square, and the probability
distribution associated with the action is shown by labels on the
edges leaving the diamond shaped square.

We say that an action $a$ is \emph{deterministic} if it moves to a single state with probability 1, i.e., if $p(a)_j = 1$ for some $j \in S$. If all the actions of a 2TBSG $G$ are deterministic we say that $G$ is deterministic.

{\bf Plays and outcomes.} A 2TBSG is played as follows. A pebble is moved from state to state starting from some initial state $i_0 \in S$. When the pebble is at state $i \in S^k$, Player $k$ chooses an action $a \in A_i$, and the pebble is moved to a random new state according to the probability distribution $p(a)$. Let $a^t$ be the $t$-th chosen action for every $t \ge 0$. The sequence of chosen actions is called a \emph{play}, and the {\em outcome} of the play, paid by Player 1 to Player 2, is $\sum_{t \ge 0} \gamma^t \cdot c(a^t)$.  

To simplify notation, we next introduce a way to represent a 2TBSG with vectors
and matrices. Figure \ref{fig:example} shows an example of this representation.

\begin{definition}[\bf Matrix representation]
Let $G = (S^1,S^2,(A_i)_{i\in S^1\cup S^2},p,c,\gamma)$ be a
2TBSG. Assume without loss of generality that $S = [n] = \{1,\dots,n\}$ and $A = [m] =
\{1,\dots,m\}$.
\begin{itemize}
\item
We define the \emph{probability matrix} $P \in
\R^{m\times n}$ by $P_{a,i} = (p(a))_i$, for all $a \in A$ and $i \in
S$.
\item
We define the \emph{cost vector} $\cc \in \R^m$ by $\cc_a = c(a)$,
for all $a \in A$.
\item
We define the \emph{source matrix} $J \in \{0,1\}^{m\times n}$ by 
$J_{a,i} = 1$ if and only if $a \in A_i$, for all $a \in A$ and $i \in
S$.
\item
We define the \emph{ownership matrix} $\I \in \{-1,0,1\}^{n\times n}$
by $\I_{i,j} = 0$ if $i \ne j$, $\I_{i,i} = -1$ if $i \in S^1$, and 
$\I_{i,i} = 1$ if $i \in S^2$.
\end{itemize}
\end{definition}

Note that $P_{a,i}$ is the probability of moving to state $i$ when
using action $a$. For a matrix $M \in \R^{m \times n}$ and a subset of
indices $B \subseteq [m]$, we let $M_B$ be the sub-matrix of $M$
consisting of rows with indices in $B$. Also, for any $i \in [m]$, we 
let $M_i \in \R^{1\times n}$ be the $i$-th row of $M$. We use similar
notation for vectors.

\begin{figure}
\center
\begin{tikzpicture}[->,>=stealth',shorten >=1pt,node distance=5.5cm*0.5,
                    semithick,scale=0.2 ]
       \colorlet{darkgray}{black!75}
\tikzstyle{every state}=[fill=white,draw=black,text=black,font=\small , inner sep=-0.05cm]
\tikzstyle{cost}=[state,regular polygon,regular polygon sides=4,regular polygon rotate=45]
\tikzstyle{max}=[state,regular polygon,regular polygon sides=4, inner sep=0.2cm]
\tikzstyle{se}=[-,shorten >=0pt];

\node [state] (s1){$1$};
\node [max] (s2)[right of=s1]{$2$};
\node [state] (s3)[right of=s2]{$3$};

\foreach \x/\y/\z in {1/7/3,2/-4/2,3/5/-10} {
\node [cost,above of=s\x] (s\x 1) {\y};
\node [cost,below of=s\x] (s\x 2) {\z};
}
\path[se] (s1) edge node[midway,left] {$a_1$} (s11);
\path[se,bend left,dashed] (s1) edge node[midway,left] {$a_2$} (s12);
\path[se,dashed] (s2) edge node[midway,left] {$a_3$} (s21);
\path[se,bend left] (s2) edge node[midway,left] {$a_4$} (s22);
\path[se] (s3) edge node[midway,left] {$a_5$} (s31);
\path[se,bend left,dashed] (s3) edge node[midway,left] {$a_6$} (s32);

\tikzstyle{every node}=[fill=white,text=black,font=\small , inner sep=0.05cm,pos=0.2]

\path (s11) edge node {$\frac{1}{2}$} (s2) 
edge[out=-15,in=130] node[pos=0.095] {$\frac{1}{2}$} (s3)
(s12) edge[bend left,dashed] (s1)
(s21) edge[dashed] (s1)
(s22) edge  node[pos=0.21] {$\frac{1}{2}$} (s1)
edge[bend left]  node {$\frac{1}{4}$} (s2)
edge  node {$\frac{1}{4}$} (s3)
(s31) edge (s2)
(s32) edge[dashed] node{$\frac{1}{3}$} (s2)
edge[bend left,dashed] node[pos=0.173]{$\frac{2}{3}$} (s3);

\end{tikzpicture}

\[
P = \begin{bmatrix}
  \vspace{0.1cm}
  0 & \frac{1}{2} & \frac{1}{2}\\
  \vspace{0.1cm}
  1 & 0 & 0\\
  \vspace{0.1cm}
  1 & 0 & 0\\
  \vspace{0.1cm}
  \frac{1}{2} & \frac{1}{4} & \frac{1}{4}\\
  \vspace{0.1cm}
  0 & 1 & 0\\
  0 & \frac{1}{3} & \frac{2}{3}
\end{bmatrix}
\quad
\cc = \begin{bmatrix}
  \vspace{0.1cm}
  7\\
  \vspace{0.1cm}
  3\\
  \vspace{0.1cm}
  -4\\
  \vspace{0.1cm}
  2\\
  \vspace{0.1cm}
  5\\
  -10
\end{bmatrix}
\quad
J = \begin{bmatrix}
  \vspace{0.1cm}
  1 & 0 & 0\\
  \vspace{0.1cm}
  1 & 0 & 0\\
  \vspace{0.1cm}
  0 & 1 & 0\\
  \vspace{0.1cm}
  0 & 1 & 0\\
  \vspace{0.1cm}
  0 & 0 & 1\\
  0 & 0 & 1
\end{bmatrix}
\quad
\I = \begin{bmatrix}
  \vspace{0.1cm}
  1 & 0 & 0\\
  \vspace{0.1cm}
  0 & -1 & 0\\
  0 & 0 & 1
\end{bmatrix}
\]

\[
\begin{array}{ccl}
  \vspace{0.1cm}
  S^1 &\!=\!& \{2\} \\
  \vspace{0.1cm}
  S^2 &\!=\!& \{1,3\} \\
  \vspace{0.1cm}
  \sigma^1 &\!=\!& \{a_4\} \\
  \vspace{0.1cm}
  \sigma^2 &\!=\!& \{a_1,a_5\} \\
  \sigma &\!=\!& \{a_1,a_4,a_5\} 
\end{array}
\quad
P_\sigma = \begin{bmatrix}
  \vspace{0.1cm}
  0 & \frac{1}{2} & \frac{1}{2}\\
  \vspace{0.1cm}
  \frac{1}{2} & \frac{1}{4} & \frac{1}{4}\\
  0 & 1 & 0\\
\end{bmatrix}
\quad
\cc_\sigma = \begin{bmatrix}
  \vspace{0.1cm}
  7\\
  \vspace{0.1cm}
  ~2~\\
  5\\
\end{bmatrix}
\quad~~
\begin{array}{|c|ccc|}
\hline
\vspace{-0.35cm}
&&& \\
t & \multicolumn{3}{c|}{\ee_1^\T P_\sigma^t}\\
\hline
\vspace{-0.35cm}
&&& \\
0 & 1 & 0 & 0 \\
\vspace{-0.35cm}
&&& \\
1 & 0 & \frac{1}{2} & \frac{1}{2} \\
\vspace{-0.35cm}
&&& \\
2 & \frac{2}{8} & \frac{5}{8} & \frac{1}{8} \\
\vspace{-0.35cm}
&&& \\
3 & \frac{10}{32} & \frac{13}{32} & \frac{9}{32} \\
\vspace{-0.5cm}
&&& \\
\vdots &  & \vdots & 
\end{array}
\]
\caption{Example of a simple 2TBSG, as well as a strategy profile $\sigma = (\sigma^1,\sigma^2)$ shown with solid lines.}
\label{fig:example}
\end{figure}

\begin{definition}[\bf Strategies and strategy profiles]
A \emph{strategy} $\sigma^k: S^k \to A^k$ for Player $k \in \{1,2\}$ maps every
state $i\in S^k$ to an action $\sigma^k(i) \in A_i$ applicable from
state $i$. A \emph{strategy profile} $\sigma = (\sigma^1,\sigma^2)$ is
a pair of strategies, one for each player. We let $\Sigma^k$ be the
set of strategies for Player $k$, and $\Sigma = \Sigma^1 \times
\Sigma^2$ be the set of strategy profiles.
\end{definition}

We view a strategy profile $\sigma = (\sigma^1,\sigma^2)$ as a map
$\sigma: S \to A$ from states to actions, such that $\sigma(i) =
\sigma^k(i)$ for all $i \in S^k$ and $k \in \{1,2\}$.

A strategy $\sigma^k \in \Sigma^k$ can also be viewed as a subset $\sigma^k
\subseteq A^k$ of actions such that $\sigma^k \cap A_i =
\{\sigma^k(i)\}$ for all $i \in S^k$. A strategy profile $\sigma =
(\sigma^1,\sigma^2) \in \Sigma$ can be viewed similarly as a subset of
actions $\sigma = \sigma^1\cup \sigma^2 \subseteq A$. Note that
$P_\sigma$ is an $n\times n$ matrix for every $\sigma \in \Sigma$. We
assume without loss of generality that actions are ordered such that $J_\sigma = I$, where
$I$ is the identity matrix, for all $\sigma \in \Sigma$. Figure
\ref{fig:example} shows a strategy profile $\sigma$ represented by
solid lines, the corresponding matrix $P_\sigma$, and the vector
$\cc_\sigma$.

The matrix $P_\sigma$ defines a Markov chain. In particular, the
probability of being in the $j$-th state after $t$ steps when starting
in state $i$ is $(P^t_\sigma)_{i,j}$. In Figure \ref{fig:example} such
probabilities are shown in the table in the lower right corner, where
$\ee_1$ is the first unit vector.
We say that the players \emph{play
according to $\sigma$} if whenever the pebble is at state $i \in S^k$,
Player $k$ uses the action $\sigma(i)$. Let $i\in S$ be some state and $t$ some number. 
The expected cost 
of the $t$-th action used when starting in state $i$ is $(P^t_\sigma)_i \cc_\sigma$. In particular, the expected outcome is
$\sum_{t=0}^\infty \gamma^t (P^t_\sigma)_i \cc_\sigma$. The following
lemma shows that this infinite series always converges.

\begin{lemma}\label{lem:converges}
For every strategy profile $\sigma\in\Sigma$ the matrix $(I-\gamma
P_\sigma)$ is non-singular, and
\[
(I-\gamma P_\sigma)^{-1} ~=~ \sum_{t=0}^\infty \gamma^t P^t_\sigma \enspace.
\]
\end{lemma}

The simple proof of Lemma~\ref{lem:converges} has been omitted. For
details we refer to, e.g., \cite{Han12}.

\begin{definition}[\bf Value vectors]\label{def:values}
For every strategy profile $\sigma \in \Sigma$ we define the
\emph{value vector} $\vv^\sigma \in \R^n$ by:
\[
\vv^\sigma ~=~ (I-\gamma P_\sigma)^{-1} \cc_\sigma \enspace.
\]
\end{definition}

The $i$-th component of the value vector $\vv^\sigma$, for a given
strategy profile $\sigma$, is the expected outcome over plays starting in $i\in S$
 when the players play
according to $\sigma$. 

It follows from Lemma~\ref{lem:converges} and
Definition~\ref{def:values} that $\vv^\sigma$ is the unique solution
to:
\begin{equation}\label{eq:one_step}
\vv^\sigma ~=~ \cc_\sigma + \gamma P_\sigma \vv^\sigma \enspace.
\end{equation}

\begin{definition}[\bf Lower and upper values]
We define the \emph{lower value vector} $\underline{\vv} \in \R^n$ and 
\emph{upper value vector} $\overline{\vv} \in \R^n$ by:
\begin{align*}
\forall i \in S: \quad \underline{\vv}_i ~&=~ \min_{\sigma^1 \in \Sigma^1} \max_{\sigma^2 \in
  \Sigma^2} ~ \vv^{(\sigma^1,\sigma^2)}_i \\
\forall i \in S: \quad 
\overline{\vv}_i ~&=~ \max_{\sigma^2 \in \Sigma^2} \min_{\sigma^1 \in
  \Sigma^1} ~ \vv^{(\sigma^1,\sigma^2)}_i \enspace.
\end{align*}
\end{definition}

Shapley~\cite{Sha53} showed that $\underline{\vv} =
  \overline{\vv}$. Hence, we may define the \emph{optimal value 
  vector} uniquely as $\vv^* := \underline{\vv} = \overline{\vv}$. 

\begin{definition}[\bf Optimal strategies]
A strategy $\sigma^1 \in \Sigma^1$ is \emph{optimal} if and only if:
\[
\forall i \in S: \quad \max_{\sigma^2 \in
  \Sigma^2} ~ \vv^{(\sigma^1,\sigma^2)}_i ~=~ \vv^*_i \enspace.
\]
Similarly, a strategy $\sigma^2 \in \Sigma^2$ is \emph{optimal} if and
only if:
\[
\forall i \in S: \quad \min_{\sigma^1 \in
  \Sigma^1} ~ \vv^{(\sigma^1,\sigma^2)}_i ~=~ \vv^*_i \enspace.
\]
We say that a strategy profile $\sigma = (\sigma^1,\sigma^2) \in
\Sigma$ is \emph{optimal} if and only if $\sigma^1$ and $\sigma^2$ are
optimal. Equivalently, $\sigma$ is a \emph{Nash equilibrium}.
\end{definition}

Note that an optimal strategy for Player 1 (Player 2) minimizes
(maximizes) the values of all states simultaneously. Hence, it is not
obvious that optimal strategies exist. This was shown, however,
by Shapley~\cite{Sha53}. \emph{Solving} a 2TBSG means finding
an optimal strategy profile, or equivalently the optimal value
vector.

\begin{definition}[\bf Reduced costs]\label{def:reduced}
For every strategy profile $\sigma \in \Sigma$ we define the vector of
\emph{reduced costs} $\bar{\cc}^\sigma \in \R^m$ by:
\[
\forall i \in S, a \in A_i: \quad
\bar{\cc}^\sigma_a ~=~ \cc_a + \gamma P_a \vv^\sigma -
\vv^\sigma_i \enspace.
\]
\end{definition}

The following theorem establishes a connection between optimal
strategies and reduced costs. For details see, e.g.,
\cite{HansenMZ13,Han12}.

\begin{theorem}[\bf Optimality condition]\label{thm:optimality}
A strategy profile $\sigma \in \Sigma$ is optimal if and only if
$(\bar{\cc}^\sigma)_{A^1} \ge 0$ and
$(\bar{\cc}^\sigma)_{A^2} \le 0$.
\end{theorem}

\subsection{LCPs for solving 2TBSGs}\label{sec:2TBSG2LCP}

Jurdzi{\'n}ski and Savani~\cite{JurdSav08} showed how the problem of
solving \emph{deterministic} 2TBSGs can be reduced to the problem of
solving $P$-matrix LCPs. We next show that the same reduction works
for general 2TBSGs.

Throughout this section we let $G = (S^1,S^2,(A_i)_{i\in
S},p,c,\gamma)$ be some 2TBSG and $(P,\cc,J,\I,\gamma)$ be the
corresponding matrix representation. We assume that there are exactly
two actions available from every state, i.e., $|A_i| = 2$ for all $i
\in S$, and we partition $A$ into two disjoint strategy profiles $\sigma$
and $\tau$.

An LCP for solving $G$ can be derived as follows.
Consider the following system of linear equations and inequalities,
where $\w,\y,\z \in \R^n$ are variables.
\begin{align}
(I-\gamma P_{\sigma}) \y - \I\w &~=~ \cc_{\sigma} \label{eq:lcp1}\\
(I-\gamma P_{\tau}) \y - \I\z &~=~ \cc_{\tau} \label{eq:lcp2}\\
\w^\T\z &~=~ 0 \label{eq:lcp3}\\
\w,\z  &~\ge~ \0 \label{eq:lcp4}
\end{align}

\begin{lemma}\label{lem:solve}
The system (\ref{eq:lcp1}), (\ref{eq:lcp2}), (\ref{eq:lcp3}), and (\ref{eq:lcp4}) has a unique solution $\w,\y,\z \in \R^n$, where $\y$ is
the optimal value vector for $G$. Also, a strategy profile $\pi$ is optimal if and only if $\pi \subseteq
\{\sigma(i) \mid i \in [n] \land \w_i = 0\} \cup \{\tau(i) \mid i \in
[n] \land \z_i = 0\}$, and such a strategy profile exists.
\end{lemma}

\begin{proof}
Observe first that if we let $\y = \vv^*$, then (\ref{eq:lcp1}) and (\ref{eq:lcp2}) ensure that $\w$ and $\z$ correspond to the reduced costs for an optimal strategy profile $\pi^*$. 
More precisely, we have:
\[
\begin{array}{rc}
\forall i\in S^1: & \bar{\cc}^{\pi^*}_{\sigma(i)} ~=~ \w_i ~~\text{and}~~ \bar{\cc}^{\pi^*}_{\tau(i)} ~=~ \z_i\\
\forall i\in S^2: & \bar{\cc}^{\pi^*}_{\sigma(i)} ~=~ -\w_i ~~\text{and}~~ \bar{\cc}^{\pi^*}_{\tau(i)} ~=~ -\z_i
\end{array}
\]
It then follows from definitions \ref{def:values} and \ref{def:reduced} that $\w^\T\z = 0$, and from Theorem \ref{thm:optimality} that $\w,\z \ge \0$. Hence the optimal value vector and the corresponding reduced costs are a solution to (\ref{eq:lcp1}), (\ref{eq:lcp2}), (\ref{eq:lcp3}), and (\ref{eq:lcp4}). We next show that for every solution $\w,\y,\z \in \R^n$, $\y$ is the optimal value vector. Since the optimal value vector is unique it follows that the system has a unique solution.

Let $\w,\y,\z \in \R^n$ be a solution to (\ref{eq:lcp1}), (\ref{eq:lcp2}), (\ref{eq:lcp3}), and (\ref{eq:lcp4}), and let $B = \{\sigma(i) \mid i \in [n] \land \w_i = 0\} \cup \{\tau(i) \mid i \in [n] \land \z_i = 0\}$. Furthermore, let $\Pi$ be the set of all strategy
profiles contained in $B$. Since $\w$ and $\z$ satisfy
(\ref{eq:lcp3}), we know that $\Pi \ne \emptyset$.

Let $a \in B \cap A_i$ for some $i \in S$.
It follows from (\ref{eq:lcp1}) and (\ref{eq:lcp2})
that $\y_i-\gamma P_{a} \y = \cc_a$. Hence, we get
from Equation (\ref{eq:one_step}) that $\y = \vv^\pi$, for every $\pi
\in \Pi$. Combining this with (\ref{eq:lcp1}), (\ref{eq:lcp2}), and
(\ref{eq:lcp4}) we get that:
\begin{align*}
\forall i \in S^1, a \in A_i: \quad 
\vv^\pi_i - \gamma P_a \vv^\pi &~\le~ \cc_a\\
\forall i \in S^2, a \in A_i: \quad 
\vv^\pi_i - \gamma P_a \vv^\pi &~\ge~ \cc_a
\end{align*}
It follows from Definition~\ref{def:reduced} and
Theorem~\ref{thm:optimality} that $\pi$ is an optimal strategy
profile.
\end{proof}

We know from Lemma~\ref{lem:converges} that $(I - \gamma P_\tau)$ is
non-singular. Hence, (\ref{eq:lcp2}) can be equivalently expressed as:
\[
\y = (I-\gamma
P_{\tau})^{-1}(\cc_{\tau}+\I\z)~.
\]
Eliminating $\y$ in
(\ref{eq:lcp1}) we then get the following equivalent equation:
\[
(I-\gamma P_{\sigma}) (I-\gamma P_{\tau})^{-1}(\cc_{\tau}+\I\z) - \I\w
~=~ \cc_{\sigma} \quad\iff
\]
\[
\I(I-\gamma P_{\sigma}) (I-\gamma P_{\tau})^{-1}\I\z - \w
~=~ \I\cc_{\sigma} - \I(I-\gamma P_{\sigma}) (I-\gamma
P_{\tau})^{-1}\cc_{\tau}
\quad\iff
\]
\begin{equation}\label{eq:complicated}
\w ~=~ (\I(I-\gamma P_{\sigma}) (I-\gamma P_{\tau})^{-1}\cc_{\tau} -
\I\cc_{\sigma}) +
\I(I-\gamma P_{\sigma}) (I-\gamma P_{\tau})^{-1}\I\z 
\end{equation}
To simplify equation (\ref{eq:complicated}) we make the following
definition.
\begin{definition}[\bf $M_{G,\sigma,\tau}$ and
    $\q_{G,\sigma,\tau}$]\label{def:solves}
Let $G$ be a 2TBSG with matrix representation $(P,\cc,J,\I,\gamma)$,
and let the set of actions of $G$ be partitioned into two disjoint
strategy profiles $\sigma$ and $\tau$. We define $M_{G,\sigma,\tau}
\in \R^{n \times n}$ and $\q_{G,\sigma,\tau} \in \R^n$ by:
\begin{align*}
M_{G,\sigma,\tau} &~=~ \I(I-\gamma P_{\sigma}) (I-\gamma P_{\tau})^{-1}\I\\
\q_{G,\sigma,\tau} &~=~ \I(I-\gamma P_{\sigma}) (I-\gamma P_{\tau})^{-1}\cc_{\tau} -
\I\cc_{\sigma} ~.
\end{align*}
\end{definition}
Equation (\ref{eq:complicated}) can now be stated as $\w =
\q_{G,\sigma,\tau} + M_{G,\sigma,\tau} \z$. 
It follows that (\ref{eq:lcp1}), (\ref{eq:lcp2}), (\ref{eq:lcp3}), and
(\ref{eq:lcp4}) can be equivalently stated as $\y = (I-\gamma
P_{\tau})^{-1}(\cc_{\tau}+\I\z)$ and:
\[
\begin{split}
\w &~=~ \q_{G,\sigma,\tau} + M_{G,\sigma,\tau} \z \\
\w^\T \z &~=~ 0 \\
\w,\z &~\ge~ \0\enspace.
\end{split}
\]
Hence, a solution to the LCP $(M_{G,\sigma,\tau},\q_{G,\sigma,\tau})$
gives a solution to  (\ref{eq:lcp1}), (\ref{eq:lcp2}),
(\ref{eq:lcp3}), and (\ref{eq:lcp4}), which, using
Lemma~\ref{lem:solve}, solves the 2TBSG $G$. We say that 
$(M_{G,\sigma,\tau},\q_{G,\sigma,\tau})$ \emph{solves} $G$.

Jurdzi{\'n}ski and Savani~\cite{JurdSav08} showed that $M_{G,\sigma,\tau}$
is a $P$-matrix when $G$ is deterministic. To prove the same for
general 2TBSGs we introduce the following lemma. The lemma is also
used in the later parts of the paper.
To understand the use of
$\vv$ in the lemma observe that $\x^\T (I-\gamma
P_{\sigma}) (I-\gamma P_{\tau})^{-1} \x = \x^\T (I-\gamma
P_{\sigma}) \vv$.

\begin{lemma}\label{lem:xj}
Let $\x$ be a non-zero vector, $\vv=(I-\gamma P_\tau)^{-1}\x$, and
$j\in \argmax_i \abs {\vv_i}$. Then:
 \begin{align}
\abs {\x_j} &~\geq~ (1-\gamma)\abs {\vv_j} \enspace
. \label{eq:xj-size}\\
\forall i: \abs {\x_i} &~\leq~ (1+\gamma) \abs {\vv_j} \enspace
.\label{eq:xi-size} \\
\x_j((I-\gamma P_{\sigma}) (I-\gamma P_{\tau})^{-1} \x)_j&~\geq~ (1-\gamma) \abs{\x_j \vv_j}  ~>~ 0 \enspace
.\label{eq:Ax-positive} 
\end{align}
\end{lemma}

\begin{proof}
Observe first that $\vv$ is the unique solution to $\vv = \x + \gamma
P_\tau \vv$. In fact, we can interpret $\vv$ as the value
vector for $\tau$ when the costs $\cc_\tau$ have been replaced by
$\x$. If $v=\0$ then this implies that $\0=\x+\0 \ne \0$ which is a
contradiction. Thus, $\vv \ne \0$ and in particular $\vv_j \ne 0$.
Since, for every $i$, the entries of $(P_\tau)_i$ are non-negative and
sum to one we have that $\abs{\gamma (P_\tau)_i \vv} \le \gamma
\abs{\vv_j}$. The equations $\vv_i =
\x_i+\gamma (P_\tau)_i \vv$, for all $i$, then imply that:
\begin{align*}
\abs{\x_j} ~&=~ \abs{\vv_j - \gamma (P_\tau)_j \vv} ~\ge~
\abs{\vv_j} - \abs{\gamma (P_\tau)_j \vv} ~\ge~
\abs{\vv_j} - \gamma\abs{\vv_j} ~=~ (1-\gamma)\abs{\vv_j} 
\enspace , \enspace \text{and}
\\
\forall i: ~ \abs{\x_i} ~&=~ \abs{\vv_i - \gamma (P_\tau)_i \vv} ~\le~
\abs{\vv_i} + \abs{\gamma (P_\tau)_i \vv} ~\le~
\abs{\vv_j} + \gamma\abs{\vv_j} ~=~ (1+\gamma)\abs{\vv_j} 
\enspace .
\end{align*}
This proves (\ref{eq:xj-size}) and (\ref{eq:xi-size}).

We next observe that $\vv_j$ and $\x_j$ have the same sign. This again
follows from $\abs{\gamma (P_\tau)_j \vv} \le \gamma \abs{\vv_j}$ and
$\vv_j = \x_j+\gamma (P_\tau)_j \vv$ .
Using that $\vv_j$ and $\x_j$ have the same sign we see that:
\begin{align*}
\x_j((I-\gamma P_{\sigma}) (I-\gamma P_{\tau})^{-1} \x)_j &~=~
\x_j((I-\gamma P_\sigma)\vv)_j ~=~ \x_j \vv_j-\gamma
\x_j(P_\sigma)_j \vv \\
&~\geq~ \x_j \vv_j-\gamma \x_j \vv_j 
~=~ (1-\gamma) \x_j \vv_j ~>~0 \enspace .
\end{align*}
This proves~(\ref{eq:Ax-positive}).
\end{proof}

We know from Lemma \ref{lem:alternative} that the matrix
$M_{G,\sigma,\tau}$ is a $P$-matrix if and only if for every $\x \neq
\0$ there exists a $j \in [n]$ such that $\x_j (M_{G,\sigma,\tau} \x)_j > 0$. 
Since $\I \x \ne \0$, inequality (\ref{eq:Ax-positive}) in Lemma
\ref{lem:xj} shows that $\x_j (M_{G,\sigma,\tau} \x)_j > 0$ for $j \in
\argmax_i \abs {((I-\gamma P_\tau)^{-1}\I\x)_i}$. Hence,
$M_{G,\sigma,\tau}$ is a $P$-matrix.

The following theorem summarizes the main result of this section.

\begin{theorem}
Let $G$ be a 2TBSG, and let $\sigma$ and $\tau$ be two disjoint
strategy profiles that form a partition of the set of actions of $G$.
Then the optimal value vector for $G$ is
$\vv^* = (I-\gamma P_{\tau})^{-1}(\cc_{\tau}+\I\z)$, where $(\w,\z)$
is a solution to the LCP
$(M_{G,\sigma,\tau},\q_{G,\sigma,\tau})$. Furthermore, $M_{G,\sigma,\tau}$ is
a $P$-matrix.
\end{theorem}

Recall that Kojima {\em et al.}~\cite{Koj91} showed that every
$P$-matrix is a $P_*$-matrix. Hence, we have shown that
$M_{G,\sigma,\tau}$ is a $P_*$-matrix.

\section{The $P_*(\kappa)$ property for 2TBSGs}\label{sec:counter}

Let $G$ be a 2TBSG with matrix representation $(P,\cc,J,\I,\gamma)$,
and let $\sigma$ and $\tau$ be two disjoint strategy profiles that
form a partition of the set of actions of $G$. Recall that $G$ can be
solved by solving the LCP $(M_{G,\sigma,\tau},\q_{G,\sigma,\tau})$. In
this section we provide essentially tight upper and
lower bounds on the smallest number $\kappa$ for which the matrix
$M_{G,\sigma,\tau}$ is guaranteed to be a $P_*(\kappa)$-matrix. More
precisely, we first show that for $\kappa=\frac{n}{(1-\gamma)^2}$, the matrix $M_{G,\sigma,\tau}$ is always a
$P_*(\kappa)$-matrix. We then show that for every
$n > 2$ and $\gamma < 1$ there exists a game $G_{n}$, and two
strategy profiles $\sigma_n$ and $\tau_n$, such that
$M_{G_{n},\sigma_n,\tau_n}$ is not a $P_*(\kappa)$-matrix for any
$\kappa < \frac{\gamma^2 (n-2)}{8(1-\gamma)^2}-\frac{1}{4}$.
It follows that the unified interior point method
of Kojima {\em et al.}~\cite{Koj91} solves the 2TBSG $G$ in time
$O(\frac{n^{4.5} L}{(1-\gamma)^2})$, where $L$ is the number of bits
required to describe $G$, and that this bound cannot be improved
further by bounding $\kappa$.

Recall that $M_{G,\sigma,\tau} = \I(I-\gamma P_{\sigma}) (I-\gamma
P_{\tau})^{-1}\I$, and define $M := \I M_{G,\sigma,\tau} \I = (I-\gamma
P_{\sigma}) (I-\gamma P_{\tau})^{-1}$. It is easy to see that
$M_{G,\sigma,\tau}$ is a $P_*(\kappa)$-matrix for some $\kappa \ge 0$
if and only if $M$ is. Indeed, the inequality of Definition
\ref{def:kappa} must hold for all $\x \in \R^n$, and we can therefore
substitute $\x$ by $\I\x$. Hence, for the remainder of this section we
will bound the $\kappa$ for which $M$ is a $P_*(\kappa)$-matrix.

\begin{theorem}\label{thm:kappa_upper_bound}
Let $G$ be a 2TBSG with $n$ states and discount factor $\gamma$, where $0<\gamma<1$. Furthermore, let $\sigma$ and $\tau$ be two strategy profiles that partition the set of actions of $G$. Then the matrix $M_{G,\sigma,\tau}$ is a $P_*(\kappa)$-matrix for
$\kappa=\frac{n}{(1-\gamma)^2}$.
\end{theorem}
\begin{proof}
As discussed above we may prove the theorem by bounding $\kappa$ for
$M = (I-\gamma
P_{\sigma}) (I-\gamma P_{\tau})^{-1}$.
We thus need to find a number $\kappa$, such that 
\[\forall \x \in \R^n:
\sum_{i\in \delta_-(M)} \!\x_i ((I-\gamma P_\sigma)(I-\gamma P_\tau)^{-1}
\x)_i~+~(1+4\kappa)\!\sum_{i\in \delta_+(M)}\! \x_i ( (I-\gamma
P_\sigma)(I-\gamma P_\tau)^{-1} \x)_i ~\geq~ 0 \enspace ,
\]  
where $\delta_-(M) = \{i \in [n] \mid \x_i(M\x)_i<0\}$ and
$\delta_+(M) = \{i \in [n] \mid \x_i(M\x)_i>0\}$.

Let $\x$ be any non-zero vector (the expression is trivially satisfied for $\x=0$), $\vv=(I-\gamma P_\tau)^{-1}\x$, and $j\in \argmax_i \abs {\vv_i}$.
To prove the lemma we estimate $\sum_{i\in \delta_-(M)} \x_i ((I-\gamma P_\sigma)\vv)_i$ and $\sum_{i\in \delta_+(M)} \x_i ((I-\gamma P_\sigma)\vv)_i$ separately. 

Using Lemma \ref{lem:xj} we see that:
\[
\forall i: ~ \abs {\x_i((I-\gamma P_\sigma)\vv)_i} ~\leq~  \abs {\x_i} (\abs
        {\vv_j}+\gamma \abs {\vv_j}) ~\leq~ (1+\gamma)^2 {\abs
          {\vv_j}}^2 ~<~ 4 {\abs
          {\vv_j}}^2 \enspace ,
\] 
which implies that:
\[\abs {\sum_{i\in \delta_-(M)} \x_i ((I-\gamma P_\sigma)\vv)_i} ~<~ 4n {\abs {\vv_j}}^2 \enspace .\]

Similarly, from Lemma \ref{lem:xj} we have that: 
\[ 
\x_j(M \x)_j ~\geq~  (1-\gamma)\x_j \vv_j  ~=~ (1-\gamma)\abs
  {\x_j}\abs{\vv_j} ~\geq~ (1-\gamma)^2 {\abs {\vv_j}}^2 \enspace ,
\]
which implies that:
\[
\sum_{i\in \delta_+(M)} \x_i ((I-\gamma P_\sigma)\vv)_i ~\geq~
(1-\gamma)^2 {\abs {\vv_j}}^2 \enspace .
\]

We conclude that:
\[
\sum_{i\in \delta_-(M)} \x_i ((I-\gamma
P_\sigma)\vv)_i ~+~(1+4\kappa)\!\!\sum_{i\in \delta_+(M)}
\x_i ((I-\gamma P_\sigma)\vv)_i ~>~ -4n
  {\abs {\vv_j}}^2+(1+4\kappa)(1-\gamma)^2 {\abs {\vv_j}}^2 \enspace .
\]
It follows that $M$ is a $P_*(\kappa)$-matrix when: 
\[
-4n {\abs {\vv_j}}^2+(1+4\kappa)(1-\gamma)^2 {\abs {\vv_j}}^2 ~\ge~ 0
\quad \iff
\]
\[
4\kappa(1-\gamma)^2 {\abs {\vv_j}}^2 ~\ge~ (4n-(1-\gamma)^2) {\abs {\vv_j}}^2
\quad \iff
\]
\[
\kappa ~\ge~ \frac{n}{(1-\gamma)^2} - \frac{1}{4} \enspace.
\]
\end{proof}

We next present a lower bound that matches the upper bound
given in Theorem~\ref{thm:kappa_upper_bound}. 
We establish the lower bound using the family of games $\{G_n \mid
n>2\}$ shown in Figure \ref{fig:high kappa}.
The game $G_n$ consists of $n$ states that all belong to Player 2. Every state has two actions, and all actions are deterministic. 
The actions from states $1$ and $2$ form self-loops.
The actions from state $i$, for $i > 2$, lead to state $1$ and $2$, respectively. 
Both actions from state $1$ have cost $1$, both actions from state $2$ have cost $-1$, and all the remaining actions have cost $a\in\R$, where $a$ will be specified in the analysis. The discount factor $\gamma < 1$ is given along with $n$.

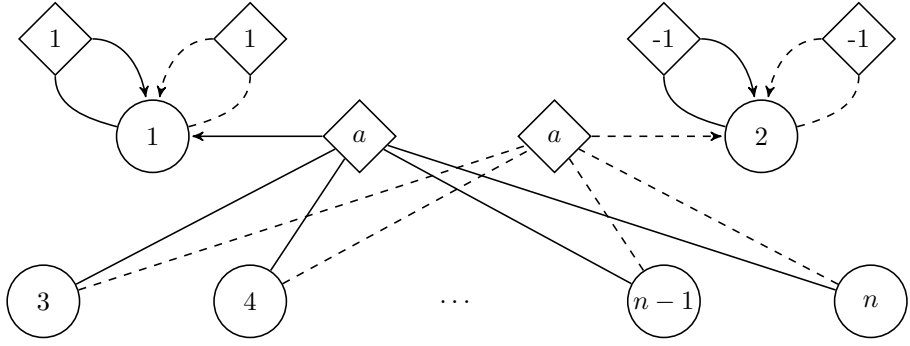
\begin{figure}
\begin{centering}
\begin{tikzpicture}[->,>=stealth',shorten >=1pt,node distance=5.5cm*0.5,
                    semithick,scale=0.2 ]
       \colorlet{darkgray}{black!75}
\tikzstyle{every state}=[fill=white,draw=black,text=black,font=\small , inner sep=-0.05cm]
\tikzstyle{cost}=[state,regular polygon,regular polygon sides=4,regular polygon rotate=45]

\node [state] (s3){$3$};
\node [state] (s4)[right of=s3]{$4$};
\node [state,draw=white] (dots)[right of=s4]{$\dots$};
\node [state] (sn1)[right of=dots]{$n-1$};
\node [state] (sn)[right of=sn1]{$n$};

\node [cost] (c1)[above left =1.6cm and 0.7cm of dots]{$a$};
\node [cost] (c2)[above right =1.6cm and 0.7cm of dots]{$a$};
\node [state,left of=c1] (s1) {1};
\node [state,right of=c2] (s2) {2};

\tikzstyle{se}=[-,shorten >=0pt];

\foreach \x/\y in {s1/1,s2/-1} {
\node [cost,above left =0.7cm and 0.7cm of \x] (\x 1) {\y};
\node [cost,above right =0.7cm and 0.7cm of \x] (\x 2) {\y};
\path (\x) edge[se,out=165,in=-90] (\x 1.south)
(\x) edge[se,out=15,in=-90,dashed] (\x 2.south)
(\x 1) edge[out=0,in=100] (\x)
(\x 2) edge[out=180,in=80,dashed] (\x);
}

\path (c1) edge (s1)
(c2) edge[dashed] (s2);
\foreach \x in {s3,s4,sn1,sn} {
\path (\x) edge[se] (c1)
(\x) edge[se,dashed] (c2);
}
\end{tikzpicture}

\end{centering}
\caption{The game $G_n$ and two strategy profiles $\sigma_n$
  (solid) and $\tau_n$ (dashed).}
\label{fig:high kappa}
\end{figure}

We also define two strategy profiles $\sigma_n$ and $\tau_n$ that partition the set of actions. Since all states belong to Player 2, $\sigma_n$ and $\tau_n$ are simply strategies for Player 2.
The strategies $\sigma_n$ and $\tau_n$ are represented in Figure \ref{fig:high
  kappa} by solid and dashed arrows, respectively. The strategies
$\sigma_n$ and $\tau_n$ both contain self-loops at states $1$ and $2$. Furthermore, for every state $i > 2$, the strategy $\sigma_n$ contains the action leading to state $1$, and the strategy $\tau_n$ contains the
action leading to state $2$. Observe that since all states belong to Player 2 we have $\I = I$ such that:
\[
M_{G_n,\sigma_n,\tau_n} = (I-\gamma P_{\sigma_n}) (I-\gamma
P_{\tau_n})^{-1} \enspace.
\]
Note also that $M_{G_n,\sigma_n,\tau_n}$ is independent of the parameter $a$.

Since all our lower bounds use $G_n$ and many of the calculations are similar, we first establish the following lemma concerning various properties of the game.

\begin{lemma}\label{lem:simple calc}
Let $n$ be given and consider $G_n$.
Let $\vv := \vv^{\tau_n} = (I-\gamma P_{\tau_n})^{-1}\cc_{\tau_n}$ be the value vector for $\tau_n$ in the game $G_n$. Define $\rr := M_{G_n,\sigma_n,\tau_n}\cc_{\tau_n}$, and let $\rr' \in \R^n$ satisfy $\rr'_i = (\cc_{\tau_n})_i \rr_i$. Then,
\begin{align*}
(\cc_{\tau_n})_i ~=~ \begin{cases}
1 & \text{if $i=1$} \\
-1 & \text{if $i=2$} \\
a & \text{if $i > 2$}
\end{cases}
&\hspace{1cm}
\vv_i ~=~ \begin{cases}
\frac{1}{1-\gamma} & \text{if $i=1$} \\
\frac{-1}{1-\gamma} & \text{if $i=2$} \\
a - \frac{\gamma}{1-\gamma} & \text{if $i > 2$}
\end{cases}\\
\rr_i ~=~ \begin{cases}
1 & \text{if $i=1$} \\
-1 & \text{if $i=2$} \\
a - \frac{2\gamma}{1-\gamma} & \text{if $i > 2$} 
\end{cases}
&\hspace{1cm}\rr'_i ~=~ \begin{cases}
1 & \text{if $i=1$} \\
1 & \text{if $i=2$} \\
a^2 - \frac{2\gamma a}{1-\gamma} & \text{if $i > 2$} \enspace .
\end{cases}
\end{align*}
\end{lemma}
\begin{proof}
The calculation of $\vv^\tau$ follows straightforwardly from (\ref{eq:one_step}), i.e., $\vv = \cc_{\tau_n} + \gamma P_{\tau_n} \vv$, which can be stated as:
\begin{align*}
\vv_{1} ~&=~ 1 + \gamma \vv_{1} \\
\vv_{2} ~&=~ -1 + \gamma \vv_{2} \\
\forall i > 2: ~~ \vv_{i} ~&=~
a + \gamma \vv_{2}
\end{align*}

Let $\rr_i = (M_{G_n,\sigma_n,\tau_n}\cc_{\tau_n})_i = ((I-\gamma P_{\sigma_n})\vv)_i = \vv_i - \gamma(P_{\sigma_n}\vv)_i$ and observe that:
\begin{align*}
\rr_{1} ~&=~ \frac{1}{1-\gamma} - \gamma \frac{1}{1-\gamma} ~=~ 1 \\
\rr_{2} ~&=~ \frac{-1}{1-\gamma} + \gamma \frac{1}{1-\gamma} ~=~ -1 \\
\forall i > 2: ~~ \rr_i ~&=~
\left(a-\frac{\gamma}{1-\gamma}\right)
-\gamma  \frac{1}{1-\gamma} ~=~
a - \frac{2\gamma}{1-\gamma} 
\end{align*}

The value of $\rr_i' = (\cc_{\tau_n})_i \rr_i$ is computed straightforwardly as well.
\end{proof}

\begin{theorem}\label{thm:lower_kappa}
Let $n > 2$ and $0<  \gamma<1$ be given. The
matrix $M_{G_n,\sigma_n,\tau_n}$ is not a $P_*(\kappa)$-matrix for
$\kappa < \frac{\gamma^2 (n-2)}{8(1-\gamma)^2}-\frac{1}{4}=\Omega\bigl(\frac{\gamma^2 n}{(1-\gamma)^2}\bigr)$.
\end{theorem}
\begin{proof}
To simplify notation we denote $M_{G_n,\sigma_n,\tau_n}$ by $M$.
Recall that $M$ is a $P_*(\kappa)$-matrix if and only if
\[\forall \x: \sum_{i\in \delta_-(M)} \!\!\x_i ((I-\gamma P_{\sigma_n})(I-\gamma P_{\tau_n})^{-1} \x)_i+(1+4\kappa)\!\!\sum_{i\in \delta_+(M)}\!\! \x_i ((I-\gamma P_{\sigma_n})(I-\gamma P_{\tau_n})^{-1} \x)_i ~\geq~ 0
\] 
where $\delta_-(M) = \{i \in [n] \mid \x_i(M\x)_i<0\}$ and
$\delta_+(M) = \{i \in [n] \mid \x_i(M\x)_i>0\}$.
We show that the inequality is violated when $\kappa < \frac{\gamma^2 (n-2)}{8(1-\gamma)^2}-\frac{1}{4}$ and $\x = \cc_{\tau_n}$ for an appropriate choice of $a$.

Let $\rr'_i = \x_i (M\x)_i$.
From Lemma~\ref{lem:simple calc} we have that \[
\rr'_i ~=~ \begin{cases}
1 & \text{if $i=1$} \\
1 & \text{if $i=2$} \\
a^2 - \frac{2\gamma a}{1-\gamma} & \text{if $i > 2$} \enspace .
\end{cases}
\]

We choose $a = \frac{\gamma}{1-\gamma}$ in order to minimize $\rr_i'$ for $i > 2$, which gives $\rr_i' = -(\frac{\gamma}{1-\gamma})^2$.
It follows that: \[\sum_{i\in \delta_-(M)} \x_i (M\x)_i ~=~ -(n-2)\left(\frac{\gamma}{1-\gamma}\right)^2 \quad\quad \text{and}\quad\quad \sum_{i\in \delta_+(M)} \x_i (M\x)_i ~=~ 2 \enspace .\]
Hence, we get:
\begin{align*}
\sum_{i\in \delta_-(M)} \x_i (M\x)_i\,+\,(1+4\kappa)\!\!\sum_{i\in \delta_+(M)} \x_i (M\x)_i &~\geq~ 0 \quad \iff \\
2(1+4\kappa) &~\ge~ (n-2) \left(\frac{\gamma}{1-\gamma}\right)^2
\quad \iff \\
\kappa &~\ge~ \frac{n-2}{8} \left(\frac{\gamma}{1-\gamma}\right)^2-\frac{1}{4} \enspace .
\end{align*}
\end{proof}

\section{Bounds for the potential reduction algorithm\label{sec:ye_alg}}

The interior point potential reduction algorithm of Kojima {\em et
  al.} \cite{KoMeYe92} for solving a $P$-matrix LCP $(M,\q)$
takes as input a parameter $\epsilon > 0$ and produces a feasible
solution $(\w,\z)$ for which $\w^\T\z<\epsilon$. Following
Ye~\cite{Ye98}, the running time of the potential reduction algorithm
is upper bounded by $O(\frac{-\delta}{\theta} n^{4} \log
\epsilon^{-1})$, where $\delta$ is the smallest eigenvalue of
$\frac{M+M^\T}{2}$, and $\theta = \theta(M)$ is the positive $P$-matrix
number of $M$ (Definition \ref{def:positive_P_matrix}). 
In this section we bound the running time of the potential reduction algorithm when
applied to 2TBSGs by studying the two quantities $\delta$ and
$\theta$.

Throughout the section we let $G$ be a 2TBSG with matrix
representation $(P,\cc,J,\I,\gamma)$, 
and $\sigma$ and $\tau$ be two disjoint strategy profiles that
form a partition of the set of actions of $G$. To simplify notation we
let $M := M_{G,\sigma,\tau}$.
We study the smallest eigenvalue $\delta$ of
$\frac{M+M^\T}{2}$ in Section \ref{sec:eigenvalue}, and the positive
$P$-matrix number $\theta(M)$ in Section \ref{sec:positive_p_matrix}. 
For both quantities we provide upper and lower bounds that are
essentially tight. 

Let $\delta_n$ be the smallest eigenvalue of the matrix $\frac{1}{2}(M_{G_n,\sigma_n,\tau_n}+M_{G_n,\sigma_n,\tau_n}^\T)$, where $M_{G_n,\sigma_n,\tau_n}$ is derived from the game in Figure \ref{fig:high kappa}. To be precise, we prove the two bounds for $\delta$ and $\delta_n$:
\begin{align*}
\delta &~>~ -\frac{(1+\gamma)\sqrt n }{1-\gamma} \\
\delta_n &~\le~ 1-\frac{\gamma \sqrt {(n-2)}}{\sqrt{2}(1-\gamma)}
\end{align*}

For the positive $P$-matrix number we prove the two bounds:
\begin{align*}
\theta(M) &~\ge~ \frac{(1-\gamma)^2}{(1+\gamma)^2 n} \\
\theta(M_{G_n,\sigma_n,\tau_n}) &~<~
\frac{(1-\gamma)^2}{(2\gamma)^2 (n-2)} \enspace.
\end{align*}

Note that the upper bounds for the smallest
eigenvalue $\delta_n$ and the positive $P$-matrix
number $\theta(M_{G_n,\sigma_n,\tau_n})$ are obtained using the same
matrix $M_{G_n,\sigma_n,\tau_n}$, which was also used in the proof of Theorem \ref{thm:lower_kappa}.
Hence, for the game $G_n$ we achieve the worst-case ratio of $\frac{-\delta}{\theta} = \Omega(\frac{\gamma  n^{3/2}}{(1-\gamma)^3})$.

\subsection{Bounds for the smallest eigenvalue}\label{sec:eigenvalue}

We first lower bound the smallest eigenvalue of $\frac{M+M^\T}{2}$,
where $M = \I(I-\gamma P_{\sigma}) (I-\gamma
P_{\tau})^{-1}\I$.
We let $\R^n_{\norm{\cdot}_\ell=1}$ be the set of vectors in $\R^n$ such that each vector $\vv\in \R^n_{\norm{\cdot}_\ell=1}$ has $\norm{\vv}_\ell=1$.

\begin{theorem}
The matrix $\frac{M+M^\T}{2}$ has smallest eigenvalue greater than $-\frac{(1+\gamma)\sqrt n }{1-\gamma}=-O(\frac{\sqrt n }{1-\gamma})$ 
\end{theorem}
\begin{proof}
Consider the equation $\lambda \x = \frac{M+M^\T}{2}\x$, where $\lambda$ is the smallest eigenvalue.  We have that: \[
\begin{split}
\lambda \x & = \frac{M+M^\T}{2}\x \\
& = \frac{\I (I-\gamma P_\sigma)(I-\gamma P_\tau)^{-1}\I+\I (I-\gamma P_\tau^\T)^{-1}(I-\gamma P_\sigma^\T) \I}{2}\x \\
& =\I  \frac{(I-\gamma P_\sigma)(I-\gamma P_\tau)^{-1}+ (I-\gamma P_\tau^\T)^{-1}(I-\gamma P_\sigma^\T) }{2} \I  \x \enspace .
\end{split}
\]
By letting $\y=\I \x$ we obtain the equation:
\[
\lambda \y = \frac{(I-\gamma P_\sigma)(I-\gamma P_\tau)^{-1}+ (I-\gamma P_\tau^\T)^{-1}(I-\gamma P_\sigma^\T) }{2}\y \enspace .
\]
We can without loss of generality assume that $\y$ has $L^2$-norm equal
to one, and by the triangle inequality we therefore have:
\[
\begin{split}
\abs \lambda & = \norm {\lambda \y}_2 = \norm {\frac{(I-\gamma P_\sigma)(I-\gamma P_\tau)^{-1}+ (I-\gamma P_\tau^\T)^{-1}(I-\gamma P_\sigma^\T) }{2}\y}_2 \\
& \leq  \frac{1}{2}\norm {(I-\gamma P_\sigma)(I-\gamma P_\tau)^{-1}\y}_2+\frac{1}{2}\norm{(I-\gamma P_\tau^\T)^{-1}(I-\gamma P_\sigma^\T)\y}_2 \enspace .
\end{split}
\]

We bound $\norm {(I-\gamma P_\sigma)(I-\gamma
  P_\tau)^{-1}\y}_2$ and $\norm{(I-\gamma
  P_\tau^\T)^{-1}(I-\gamma P_\sigma^\T)\y}_2$ separately. 
We first observe that: \[
\begin{split}
\norm {(I-\gamma P_\sigma)(I-\gamma P_\tau)^{-1}\y}_2 & ~\leq~  \max_{\vv\in \R^n_{\norm{\cdot}_\infty=1}} \norm {(I-\gamma P_\sigma)(I-\gamma P_\tau)^{-1}\vv}_2 \\
& ~=~ \max_{\vv\in \R^n_{\norm{\cdot}_\infty=1}} \norm { (I-\gamma P_\sigma)\sum_{t=0}^\infty \gamma^t P_\tau^t \vv}_2 \\
& ~\leq~ \max_{\vv\in \R^n_{\norm{\cdot}_\infty=1}} \norm { (I-\gamma P_\sigma)\sum_{t=0}^\infty \gamma^t \vv}_2 \\
& ~=~ \max_{\vv\in \R^n_{\norm{\cdot}_\infty=1}} \norm { (I-\gamma P_\sigma)\frac{\vv}{1-\gamma}}_2 \\
& ~\le~ \max_{\vv\in \R^n_{\norm{\cdot}_\infty=1}} \norm { \frac{(1+\gamma)\vv}{1-\gamma}}_2 \\
& ~=~ \frac{1+\gamma}{1-\gamma}\max_{\vv\in \R^n_{\norm{\cdot}_\infty=1}}  \norm {\vv}_2 \\
& ~=~ \frac{(1+\gamma)\sqrt{n}}{1-\gamma} \enspace .
\end{split}
\]
Here, the first inequality follows from the fact that if 
$\vv\in \R^n$ has $L^2$-norm equal to 1, then it has $L^\infty$-norm equal
to at most 1. The first equality follows from Lemma \ref{lem:converges}.
To prove the second inequality we use that $\norm{P_\tau^t \vv}_\infty \le
\norm{\vv}_\infty$, for all $t \ge 0$, since the entries of $P_\tau$ are in $[0,1]$.
The third inequality follows from the fact that $\norm{(I-\gamma
  P_\sigma)\vv}_\infty \le (1+\gamma)\norm{\vv}_\infty$.
The last equality follows from the fact that if a vector $\vv\in \R^n$ has $L^\infty$-norm equal to $1$ then it has $L^2$-norm at most $\sqrt{n}$.

We also have that: \[
\begin{split}
\norm {(I-\gamma P_\tau^\T)^{-1}(I-\gamma P_\sigma^\T)\y}_2  & ~\leq~ \max_{\vv\in \R^n_{\norm{\cdot}_1=1}} \norm { (I-\gamma P_\tau^\T)^{-1}(I-\gamma P_\sigma^\T) \sqrt n \vv}_2 \\
& ~\leq~ (1+\gamma) \sqrt  n \max_{\vv\in \R^n_{\norm{\cdot}_1=1}} \norm { (I-\gamma P_\tau^\T)^{-1} \vv}_2 \\
& ~=~ (1+\gamma) \sqrt  n \max_{\vv\in \R^n_{\norm{\cdot}_1=1}} \norm { \sum_{t=0}^\infty \gamma^t (P_\tau^t)^\T \vv}_2 \\
& ~\leq~ (1+\gamma) \sqrt n \max_{\vv\in \R^n_{\norm{\cdot}_1=1}} \norm{  \sum_{t=0}^\infty \gamma^t \vv}_2 \\
& ~=~ (1+\gamma) \sqrt n \max_{\vv\in \R^n_{\norm{\cdot}_1=1}} \norm { \frac{\vv }{1-\gamma}}_2\\
& ~=~ \frac{(1+\gamma) \sqrt n}{1-\gamma} \max_{\vv\in \R^n_{\norm{\cdot}_1=1}} \norm { \vv}_2\\
& ~=~ \frac{(1+\gamma) \sqrt n}{1-\gamma} \enspace .
\end{split}
\]
Here, the first inequality follows from the fact that if a vector
$\vv\in \R^n$ has $L^2$-norm equal to 1 then it has $L^1$-norm equal to at
most $\sqrt n$. The second inequality follows from the fact that the
columns of $P_\sigma^\T$ sum to 1 such that $\norm{(I-\gamma
  P_\sigma^\T) \vv}_1 \le (1+\gamma) \norm{\vv}_1$.
The first equality follows from Lemma \ref{lem:converges}.
For the third inequality we again use that the
columns of $P_\sigma^\T$ sum to 1, which implies that
$\ee^\T P_\sigma^\T \vv = \ee^\T \vv$ such that
$\norm{P_\sigma^\T \vv}_1 \le \norm{\vv}_1$.
The last equality follows from the fact that if a vector $\vv\in
\R^n$ has $L^1$-norm equal to 1 then it has $L^2$-norm at most 1.

Hence,
\[
\abs \lambda ~\leq~ \frac{1}{2}\frac{(1+\gamma) \sqrt
  n}{1-\gamma}+\frac{1}{2}\frac{(1+\gamma) \sqrt n}{1-\gamma} ~=~ \frac{(1+\gamma) \sqrt n}{1-\gamma} \enspace ,
\]
which completes the proof.
\end{proof}

We next upper bound the smallest eigenvalue of
$\frac{M_{G_n,\sigma_n,\tau_n}+M^\T_{G_n,\sigma_n,\tau_n}}{2}$, where
$G_n$, $\sigma_n$, and $\tau_n$ are defined in Section
\ref{sec:counter} (Figure \ref{fig:high kappa}).

\begin{theorem}
Let $n>2$ and $0<\gamma<1$ be given, and let $M :=
M_{G_n,\sigma_n,\tau_n}$. The matrix $\frac{M+M^\T}{2}$ has smallest
eigenvalue at most $1-\frac{\gamma \sqrt {n-2}}{\sqrt{2}(1-\gamma)}$. 
\end{theorem}

\begin{proof}
Recall that
\[
(\cc_\tau)_i ~=~ \begin{cases}
1 & \text{if $i=1$}\\
-1 & \text{if $i=2$}\\
a & \text{if $i>2$~.}
\end{cases} 
\]
We show that, for an appropriate choice of $a$, $\x=\cc_\tau$ is an eigenvector for the matrix $\frac{M+M^\T}{2}$ with eigenvalue $\lambda = 1-\frac{\gamma \sqrt {n-2}}{\sqrt{2}(1-\gamma)}$.
Hence, we prove that the following equation is satisfied:
\[
\begin{split}
\lambda \x & ~=~ \frac{M+M^\T}{2}\x \\
& ~=~ \frac{\I (I-\gamma P_\sigma)(I-\gamma P_\tau)^{-1}\I+\I (I-\gamma P_\tau^\T)^{-1}(I-\gamma P_\sigma^\T) \I}{2}\x \enspace .
\end{split}
\]
Recall that $\I=I$. We evaluate the two terms on the right-hand-side separately and from right to left.

We first evaluate $\rr:=(I-\gamma P_\sigma)(I-\gamma P_\tau)^{-1}\x$. 
From Lemma~\ref{lem:simple calc} we get that
\begin{align*}
\rr_{i} ~&=~\begin{cases}
1 & \text{if $i=1$} \\
-1 & \text{if $i=2$} \\
a-\frac{2\gamma}{1-\gamma}& \text{if $i > 2$} \enspace .
\end{cases}
\end{align*}

We next evaluate $
\rr'':=(I-\gamma P_\tau^\T)^{-1}(I-\gamma P_\sigma^\T)\x$.
Let $\vv '=(I-\gamma P_\sigma^\T)\x$.  Observe that no actions move to
state $i$ for $i > 2$, which means that $(P_\sigma^\T)_i=\0$ and $(P_\tau^\T)_i=\0$ for $i > 2$. By similarly considering the incoming actions for states 1 and 2, it is easy to check that
\[
\vv'_i~=~\begin{cases}
1 -\gamma - \gamma a(n-2) & \text{if $i=1$}\\
-1+\gamma & \text{if $i=2$} \\
a & \text{if $i>2$} 
\enspace .
\end{cases}
\]
Since $\rr''= (I-\gamma P_\tau^\T)^{-1}\vv'$ we see that
 $\rr''_i=\vv'_i+\gamma(P_\tau^\T)_i \rr''$. For $i>2$, it follows that $\rr''_i=\vv'_i+0=a$. For $i=1$, only the self-loop moves to state $1$ in $\tau$, and we get that
\[
\rr''_1~=~\vv'_1+\gamma\rr''_1 ~=~ 
1 -\gamma - \gamma a(n-2) +\gamma\rr''_1 ~=~
1-\frac{\gamma a(n-2)}{1-\gamma}\enspace .
\]
For $i=2$, all actions in $\tau$, except the one from state 1, move to state 2. We therefore see that 
\[
\rr''_2~=~\vv'_2+\gamma\sum_{j=2}^n\rr''_j~=-1+\gamma+\gamma a(n-2)+\gamma\rr''_2~=~-1+\frac{\gamma a(n-2)}{1-\gamma}\enspace .
\]
Hence, we have that
\[
\rr''_i~=~\begin{cases}
1-\frac{\gamma a(n-2)}{1-\gamma}& \text{if $i=1$}\\
-1+\frac{\gamma a(n-2)}{1-\gamma}& \text{if $i=2$} \\
a & \text{if $i>2$} 
\enspace .
\end{cases}
\]
 
We conclude that
\[
\left(\frac{M+M^\T}{2}\x\right)_i~=~\frac{\rr_i+\rr''_i}{2}~=~\begin{cases}
1-\frac{\gamma a(n-2)}{2(1-\gamma)}  & \text{if $i=1$} \\
-1+\frac{\gamma a(n-2)}{2(1-\gamma)}  & \text{if $i=2$} \\
a-\frac{\gamma}{1-\gamma} & \text{if $i>2$} \enspace .
\end{cases}
\]

It remains to show that $\lambda \x=\frac{\rr+\rr''}{2}$, which corresponds to the system of equations: \begin{align*}
\lambda &=1-\frac{\gamma a(n-2)}{2(1-\gamma)} \\ 
\lambda a &=a-\frac{\gamma}{1-\gamma}
\enspace .
\end{align*}
By eliminating $\lambda$, using the first equation, we get: 
\begin{align*}
a-\frac{\gamma a^2(n-2)}{2(1-\gamma)}~&=~ a-\frac{\gamma}{1-\gamma}\quad
\iff \\
\frac{\gamma}{1-\gamma} ~&=~\frac{\gamma a^2(n-2)}{2(1-\gamma)} \quad
\iff \\
2~&=~a^2(n-2) \quad \iff \\
a ~&=~\pm \sqrt \frac{2}{n-2} \enspace .
\end{align*}
Thus for $a=\sqrt \frac{2}{n-2}$ we get that 
\[
\lambda =1-\frac{\gamma \sqrt {n-2}}{\sqrt{2}(1-\gamma)}
\]
is an eigenvalue for $\frac{M+M^\T}{2}$ with eigenvector $\x=\cc_\tau$, and in particular the smallest eigenvalue for $\frac{M+M^\T}{2}$ is at most $1-\frac{\gamma \sqrt {n-2}}{\sqrt{2}(1-\gamma)}$.
\end{proof}

\subsection{Bounds for the positive $P$-matrix number}\label{sec:positive_p_matrix}

We next lower bound the positive $P$-matrix number for any 2TBSG.

\begin{theorem}
Let $n$ and $0<\gamma<1$ be given. For any  2TBSG $G$ with $n$ states,
the matrix $M := M_{G,\sigma,\tau}$, where $\sigma$ and $\tau$ partition
the actions of $G$, has positive $P$-matrix number, $\theta(M)$, at least $\frac{(1-\gamma)^2}{(1+\gamma)^2 n}=\Omega(\frac{(1-\gamma)^2}{n})$ 
\end{theorem}

\begin{proof}
Recall that the positive $P$-matrix number of $M = M_{G,\sigma,\tau}$ is
defined as:
\[
\theta(M) ~=~ \min_{\norm{\x}_2 =1}~\max_{i\in[n]} ~~ \x_i (M\x)_i \enspace. 
\]

Let $\x \in \R^n_{\norm{\cdot}_2=1}$ be given. Let  $\vv=(I-\gamma
P_\tau)^{-1}\x$ and $j\in \argmax_i \abs {\vv_i}$.
From Lemma \ref{lem:xj} we know that $\x_j(M
\x)_j \geq (1-\gamma) \abs{\x_j \vv_j} \ge (1-\gamma)^2 (\vv_j)^2$. We also
know from Lemma \ref{lem:xj} that $\vv^2_j \ge \frac{\x_i^2}{(1+\gamma)^2}$ for all $i \in [n]$. Hence, we see that
$\x_j(M \x)_j \geq \frac{(1-\gamma)^2 \x_i^2}{(1+\gamma)^2}$ for all
$i \in [n]$. Since $\norm{\x}_2 = 1$ there must exist an index $i$
such that $\abs{\x_i} \ge \frac{1}{\sqrt{n}}$. It follows that 
$\x_j(M \x)_j \geq \frac{(1-\gamma)^2}{(1+\gamma)^2 n}$. Since this
inequality holds for all $\x \in \R^n_{\norm{\cdot}_2=1}$ we see that 
$\theta(M) \ge \frac{(1-\gamma)^2}{(1+\gamma)^2 n}$.

\end{proof}

We next upper bound the positive $P$-matrix number of
$M_{G_n,\sigma_n,\tau_n}$, where once again $G_n$, $\sigma_n$, and $\tau_n$ are from the construction shown in Figure \ref{fig:high kappa}.

\begin{theorem}
Let $n>2$ and $0< \gamma<1$ be given. The matrix $M:=
M_{G_n,\sigma_n,\tau_n}$ has positive $P$-matrix number $\theta(M) <
\frac{(1-\gamma)^2}{(2\gamma)^2 (n-2)}$.
\end{theorem}

\begin{proof}
Recall that the positive $P$-matrix number of $M = M_{G_n,\sigma_n,\tau_n}$ is
defined as:
\[
\theta(M) ~=~ \min_{\norm{\x}_2 =1}~\max_{i\in[n]} ~~ \x_i (M\x)_i \enspace. 
\]
Recall also that
\[
(\cc_\tau)_i ~=~ \begin{cases}
1 & \text{if $i=1$}\\
-1 & \text{if $i=2$}\\
a & \text{if $i>2$~.}
\end{cases} 
\]

Define $\x := \cc_\tau$ and $\x' := \frac{\x}{\norm{\x}_2}$.
We show that, for an appropriate choice of $a$, $\x'$ satisfies
\[
\max_{i\in[n]} ~~ \x'_i (M\x')_i~<~\frac{(1-\gamma)^2}{(2\gamma)^2 (n-2)}~.
\]
This implies that $\theta(M)<\frac{(1-\gamma)^2}{(2\gamma)^2 (n-2)}$, completing the proof.

Let $\rr' \in \R^n$ satisfy $\rr'_i = \x_i (M\x)_i$ for all $i \in [n]$.
From Lemma~\ref{lem:simple calc}, we then get that
\[(\rr')_i ~=~ \x_i (M\x)_i ~=~ \begin{cases}
1 & \text{if $i=1$} \\
1 & \text{if $i=2$} \\
a^2 - \frac{2\gamma a}{1-\gamma} & \text{if $i > 2$} \enspace .
\end{cases}
\]
Since $\norm{\x}_2=\sqrt{2+(n-2)a^2}$ it follows that
\[ 
\x'_i (M \x')_i~=~ \frac{\x_i (M \x)_i}{\norm{\x}_2^2} ~=~ \begin{cases}
\frac{1}{2+(n-2)a^2} & \text{if $i=1$} \\
\frac{1}{2+(n-2)a^2} & \text{if $i=2$} \\
\frac{a^2 - \frac{2\gamma a}{1-\gamma}}{2+(n-2)a^2} & \text{if $i > 2$} \enspace ,
\end{cases}
\]
and for $a = \frac{2\gamma}{1-\gamma}$ we see that:
\[
\x'_i (M \x')_i ~=~ 
\begin{cases}
\frac{1}{2+(n-2)\frac{(2\gamma)^2}{(1-\gamma)^2}} & \text{if $i=1$} \\
\frac{1}{2+(n-2)\frac{(2\gamma)^2}{(1-\gamma)^2}} & \text{if $i=2$} \\
0 & \text{if $i > 2$} \enspace .
\end{cases}
\]

Since
\begin{align*}
\frac{1}{2+(n-2)\frac{(2\gamma)^2}{(1-\gamma)^2}} ~<~ \frac{1}{(n-2)\frac{(2\gamma)^2}{(1-\gamma)^2}} ~=~\frac{(1-\gamma)^2}{(2\gamma)^2(n-2)} \enspace ,
\end{align*}
there is a vector $\x'$, with $\norm{\x'}_2=1$, such that 
\[\max_{i\in[n]}~~ \x'_i (M\x')_i~<~\frac{(1-\gamma)^2}{(2\gamma)^2(n-2)}\enspace ,\] implying that \[
\theta(M) ~=~ \min_{\norm{\x}_2 =1}~\max_{i\in[n]} ~~ \x_i (M\x)_i ~<~\frac{(1-\gamma)^2}{(2\gamma)^2(n-2)} \enspace. 
\]
\end{proof}
\bibliographystyle{abbrv}
\bibliography{../bibliography}

\end{document}